\documentclass[12pt]{article}
\usepackage{amsmath,amssymb,amsthm,natbib,url,graphicx,setspace,booktabs,multirow,color}
\usepackage[margin = 1in]{geometry}

\newcommand{\E}[1]{\mathbb{E}\left(#1\right)}
\newcommand{\EC}[2]{\mathbb{E}\left(#1\big|#2\right)}
\newcommand{\V}[1]{\mathbb{V}\left(#1\right)}
\newcommand{\diff}{\,\mathrm{d}}
\newcommand{\prob}{\mathbb{P}}
\newcommand{\ind}{\mathbf{1}}

\newlength{\fixboxwidth}
\theoremstyle{plain}  
\newtheorem{thm}{Theorem}[section] 
\newtheorem{lem}[thm]{Lemma} 
\newtheorem{prop}[thm]{Proposition} 
\newtheorem*{cor}{Corollary} 

\theoremstyle{definition} 
\newtheorem{defn}{Definition}[section] 
 
\newtheorem{exmp}{Example}[section]

\theoremstyle{remark}

\title{Generic Conditions for Forecast Dominance\thanks{We thank Werner Ehm, Tilmann Gneiting, Malte Kn\"uppel, Sebastian Lerch, Melanie Schienle and seminar and conference participants at the Heidelberg Institute for Theoretical Studies, Statistische Woche (Linz, 2018), the University of Cologne and QFFE (Marseille, 2019) for helpful discussions, and Yoon-Jae Whang for making program code related to \cite{LintonEtAl2005} available via his website. Johanna F.~Ziegel gratefully acknowledges financial support from the Swiss National Science Foundation.}}
\author{Fabian Kr\"uger\thanks{Karlsruhe Institute of Technology, Department of Economics, Bl\"ucherstra{\ss}e 17, 76185 Karlsruhe, Germany. Email: \texttt{fabian.krueger@kit.edu}}\and Johanna F.~Ziegel\thanks{University of Bern, Institute of Mathematical Statistics and Actuarial Science, Alpeneggstrasse 22, 3012 Bern, Switzerland. Email: \texttt{johanna.ziegel@stat.unibe.ch}.}}
\date{December 14, 2019}

\begin{document}

\doublespacing
\maketitle

\begin{abstract}
Recent studies have analyzed whether one forecast method dominates another under a class of consistent scoring functions. While the existing literature focuses on empirical tests of forecast dominance, little is known about the theoretical conditions under which one forecast dominates another. To address this question, we  derive a new characterization of dominance among forecasts of the mean functional. We present various scenarios under which dominance occurs. Unlike existing results, our results allow for the case that the forecasts' underlying information sets are not nested, and allow for uncalibrated forecasts that suffer, e.g., from model misspecification or parameter estimation error. We illustrate the empirical relevance of our results via data examples from finance and economics.
\end{abstract}

\section{Introduction}

Forecasts of a random variable $Y$ (such as the inflation rate, a financial volatility measure, or the sale price of a house) play an important role in economics. Recent technological advances have contributed to an ever increasing array of data sources and forecasting techniques, which necessitates statistically principled comparisons of forecast quality. Here we focus on the typical task of predicting the mean of $Y$. It is well known that squared error loss sets the incentive to  correctly forecast the mean, conditional on a certain information set. This basic insight underlies the use of squared error for estimating regression models. However, \cite{Savage1971} shows that there are infinitely many other scoring (or loss) functions that are also consistent with the goal of forecasting the mean. Consider, for example, the task of modeling and forecasting the mean of a binary variable $Y \in \{0, 1\}$, which is simply the probability that $Y = 1$. In this case, squared error is often referred to as the `Brier score' \citep[following][]{Brier1950}. While squared error can be used to construct consistent parameter estimators in regression models and to evaluate probability forecasts out-of-sample, there is a continuum of other scoring functions that can be used as well \citep[see e.g.][]{BujaEtAl2005}. The Bernoulli log likelihood function, which corresponds to maximum likelihood estimation, is arguably the most popular of these choices. In the general case where $Y$ is not restricted to be binary, squared error continues to be a popular scoring function, and can be motivated as the (negative) log likelihood function of a Gaussian density with known variance. Log likelihood functions corresponding to other single-parameter families (such as Poisson or Exponential) can be employed as well; Table \ref{tab:examples} below provides examples.

The non-uniqueness of consistent scoring functions is challenging, in that rankings of two forecast methods by average scores may depend on the specific function used for out-of-sample evaluation. \cite{Ehm2016}, \cite{EhmKruger2017}, \cite{YenYen2018}, \cite{ZiegelEtAl2017} and \cite{BarendsePatton2019} therefore propose graphical tools and hypothesis tests to analyze the robustness of empirical forecast rankings. In their terminology, one forecast method dominates another if it performs better in terms of every consistent scoring function.

Adopting a theoretical perspective, \cite{Holzmann2014} show that a correctly specified forecast method dominates a competitor that is based on a smaller (nested) information set. However, forecasts based on diverse and thus non-nested information sets play a major role in applications, and are often encouraged by designers of forecast surveys and contests. For example, the European Central Bank's `Survey of Professional Forecasters' features private and public-sector, financial and non-financial institutions from all over Europe \citep{ECB2018}. \cite{Patton2017} demonstrates that non-nested information sets may lead to lack of forecast dominance, i.e., to forecast rankings that fail to be robust across consistent scoring functions. This issue has been tackled for probability forecasts of a binary variable \citep{DegrootFienberg1983,KrzysztofowiczLong1990}, but results for more general situations are available only under specific assumptions. Furthermore, all existing theoretical results assume that the forecasts under comparison specify the correct expectation of $Y$, given some information set. As illustrated by \cite{Patton2017}, this assumption is often violated in applications, which may lead to non-robust forecast rankings.

The present paper sheds new light on the theoretical conditions under which forecast dominance occurs. An understanding of these conditions is useful to interpret empirical results of (non-)robust forecast rankings, and to identify desiderata of forecasting methods that may inspire improvements of existing methods. Unlike previous studies, we derive conditions that allow for non-nested information sets. Furthermore, we allow for various types of forecast imperfections resulting,  amongst others, from model misspecification and parameter estimation error (if forecasts are generated by statistical methods) or cognitive biases (if forecasts are judgmental, generated by humans). These phenomena are ubiquitous in practice but have not been tackled by the existing theoretical literature on forecast dominance.

The paper is structured as follows. Section \ref{sec:mean} presents our main technical result, a new characterization of dominance among mean forecasts. We then discuss alternative sets of assumptions that yield natural conditions for dominance. Section \ref{sec:auto} considers the case of auto-calibrated forecasts, which means that the forecast matches the conditional expectation of $Y$, given the forecast itself. Under this condition, which allows for non-nested information sets, the forecast which is more variable in the sense of convex order \citep[see e.g.][]{ShakedShanthiku2007,Levy2016} dominates the other. This result generalizes the result of \cite{Holzmann2014} mentioned above, and thus provides weaker sufficient conditions for forecast dominance. Section \ref{sec:gauss} drops the auto-calibration assumption, but instead requires joint normality of each forecast with the predictand. Alternatively, Section \ref{sec:model} assumes that both forecasts are based on the same information set $\mathcal{F}$, but yield imperfect approximations of the conditional expectation of the predictand given $\mathcal{F}$. Our results in Sections \ref{sec:gauss} and \ref{sec:model} demonstrate that there can well be dominance relations among two uncalibrated (i.e., not auto-calibrated) forecasts. In Section \ref{sec:data}, we illustrate our theoretical results via data examples from finance and economics. Section \ref{sec:disc} concludes with a discussion of the results and open problems. All proofs are deferred to the appendix. An online appendix presents additional analytical examples and details on hypothesis testing in our data examples.

\section{A Characterization of Forecast Dominance}
\label{sec:mean}

\cite{Savage1971} considers scoring functions of the form
\begin{equation}
S(x, y) = \phi(y) - \phi(x) - \phi'(x)~(y-x),\label{csf}
\end{equation}
where $x \in \mathbb{R}$ is a forecast, $y\in\mathbb{R}$ is a realization, and $\phi$ is a convex function with subgradient $\phi'$. Here, a scoring function assigns a negatively oriented penalty, such that a smaller value of $S$ corresponds to a better forecast. Functions of the form given in \eqref{csf} are \emph{consistent} for the mean \citep{Gneiting2011}: If $Y$ has cumulative distribution function (CDF) $F$, then
\begin{equation}
\E{S(m(F), Y)} \le \E{S(x, Y)}, \quad \text{for any $x \in \mathbb{R}$.}\label{consistent}
\end{equation}
Here $m(F)=\int x\diff F(x)$ is the mean of $F$ (which we always assume to exist and be finite), and $\mathbb{E}$ denotes expectation. Equation (\ref{consistent}) states that a forecaster minimizes their expected score when stating the mean of $Y$ as their forecast. The scoring function $S$ is \emph{strictly consistent} if equality in (\ref{consistent}) implies $x = m(F)$. Strict consistency corresponds to a strictly convex function $\phi$ in (\ref{csf}). Under some additional assumptions \citep[see][Theorem 7]{Gneiting2011}, the scoring functions given at \eqref{csf} are the only consistent scoring functions for the mean. Note that the additive term $\phi(y)$ in (\ref{csf}) is included to enforce the convention that $S(y,y) = 0$. However, the term does not depend on $x$, and is hence irrelevant in terms of optimal forecasting. Table \ref{tab:examples}, which is a modified version of \citet[Table 1]{YenYen2018}, presents examples of strictly consistent scoring functions for the mean.

\begin{table}[!htbp]
	\centering
	\footnotesize
	\begin{tabular}{cclll}
		&& Range & Range & \\
		$S(x,y)$ & $\phi(z)$ &  of $X$ &  of $Y$  & Comment(s) \\ \toprule
		$(y-x)^2$ & $z^2$ & $\mathbb{R}$ & $\mathbb{R}$ &  squared error \\ [.4cm]
		\scriptsize{$-y \log x - (1-y)\log (1-x)^{*}$} & \scriptsize{$z \log z + (1-z)\log(1-z)$}& $(0,1)$ & $[0,1]$ & negative log likelihood of \\ &&&& Bernoulli dist.\\[.4cm]
		$\log x + \frac{y}{x} - 1^{*}$& $-\log z$ & $(0,\infty)$ & $[0, \infty)$ & negative log likelihood of\\ &&&& exponential dist.; equal to \\
		&&&& QLIKE loss \citep{Patton2011} \\ [.4cm]
		$-y \log {x} + x^{*}$ & $z\log z - z$ & $(0,\infty)$ & $[0, \infty)$ &   negative log likelihood of\\ &&&& Poisson dist.\\ \bottomrule
	\end{tabular}
	\caption{\small Examples of strictly consistent scoring functions for the mean. Each example is characterized by a strictly convex function $\phi(z)$. Scoring functions marked by an asterisk  ($^*$) differ from Equation (\ref{csf}) by subtracting $\phi(y)$. This transformation ensures that the scoring function is well-defined over the entire range of $Y$. Rankings of any two forecasts $x_1, x_2$ remain unchanged, and strict consistency of the scoring function is preserved.\label{tab:examples}}
\end{table}

Consider two generic forecasters (or forecasting methods) A and B who issue forecasts $X_A$ and $X_B$ of the mean of $Y$. We treat these forecasts as random variables and consider their joint distribution with $Y$, the random variable to be predicted. We assume throughout that $X_A$, $X_B$ and $Y$ are integrable. The random variables are defined on the probability space $(\Omega, \mathcal{A}, \mathbb{Q})$ whereby the point forecasts $X_A, X_B$ are measurable with respect to information sets $\mathcal{A}_A, \mathcal{A}_B \subseteq \mathcal{A}$; see \citet[Section 3.1]{Ehm2016} for a detailed discussion. This setup includes the case of a binary predictand $Y \in \{0, 1\}$, in which the mean forecasts $X_A, X_B$ quote the probability that $Y = 1$, conditional on their respective information sets. We emphasize that the setup is consistent with the case that $Y \equiv Y_t$ is a time series and $X_j \equiv X_{tj}, j \in \{A, B\}$ are associated forecasts. The only requirement is that the joint distribution of the forecasts and the predictand is strictly stationary, such that the objects that we use in the following (notably expectations and CDFs) are well defined and do not depend on time. See \citet[Definition 2.2]{StraehlZiegel2017} for a formal probability space setup involving time series of forecasts and realizations, and Example \ref{exmp:3} for an illustration. The following notion of forecast dominance is central to this paper.

\begin{defn}[Forecast dominance]\label{def:dom}
Forecast $A$ \emph{dominates} forecast $B$ if $$\E{S(X_A, Y)} \le \E{S(X_B, Y)}$$
for every function $S$ of the form given in (\ref{csf}).
\end{defn}

The preceding definition implies that the better performance of $A$ compared to $B$ is robust across all consistent scoring functions $S$. {Theorem 1b and Corollary 1b of \citet{Ehm2016} imply that forecast dominance holds if and only if $\E{S_\theta(X_A, Y)} \le \E{S_\theta(X_B, Y)}$ for all $\theta \in \mathbb{R}$,
where 
\begin{align}
S_\theta(x,y) &= \frac{1}{2}(\theta-y)\mathbf{1}_{(x>\theta)}\label{esmean}
\end{align}
is the so-called elementary score for the mean indexed by the parameter $\theta \in \mathbb{R}$, up to a term that does not depend on $x$ and is thus irrelevant in terms of forecast rankings (see Lemma \ref{lem:A.3} for details). Building upon the elementary score}, we next present a novel characterization of forecast dominance. 

\begin{thm}\label{thm:mean}
	Let $A$ and $B$ be forecasts for the mean. Then $A$ dominates $B$ if and only if $\psi_{A}(\theta) \ge \psi_B(\theta)$ for all $\theta \in \mathbb{R}$,
	where 
	\begin{align*}
	\psi_j(\theta) &=\frac{1}{2}\int_{\theta}^\infty \prob(X_j > w)\diff w + \frac{1}{2} \E{(\EC{Y}{X_j}-X_j)\mathbf{1}_{(X_j > \theta)}} \quad \text{for $j \in \{A,B\}$.}
	\end{align*}
\end{thm}

{The function $\psi_j(\theta)$ appearing in Theorem \ref{thm:mean} is the expected value of the random variable $-S_\theta(X_j,Y)$, where $S_\theta(x,y)$ has been defined at (\ref{esmean}). Theorem \ref{thm:expectiles} in Appendix \ref{app:A} is a more general version of Theorem \ref{thm:mean}, covering forecast dominance for expectiles at level $\tau \in (0,1)$. Expectiles are an asymmetric generalization of the mean which is the expectile at level $\tau = 1/2$ \citep{NeweyPowell1987}. While the representation of forecast dominance in \cite{Ehm2016} is an important prerequisite for our Theorems \ref{thm:expectiles} and \ref{thm:mean}, our derivation of an analytical expression for the {expected} score is novel, and is crucial in order to establish forecast dominance (or lack thereof) in theoretical scenarios. 
The two summands of the function $\psi_j$ separate the influence of the variability (first summand) and the calibration (second summand) of the forecast. Roughly speaking, calibration refers to the statistical compatibility of forecasts and observations; see Section \ref{sec:auto} for details. Variability of a forecast may or may not be desirable depending on the calibration properties; see Theorem \ref{thm:1} and Proposition \ref{prop:bivnormal}.} In the remainder of this paper, we derive various interpretable scenarios under which the technical condition of Theorem \ref{thm:mean} is satisfied. 

\section{Auto-Calibrated Forecasts}\label{sec:auto}

\begin{defn}[Auto-calibration]\label{def:auto}
$X$ is an \emph{auto-calibrated} forecast of $Y$ if $\EC{Y}{X} = X$ almost surely.
\end{defn}

The definition implies that the forecast $X$ of $Y$ can be used `as is', without any need to perform bias correction. The prefix `auto' indicates that $X$ is an optimal forecast relative to the information set $\sigma(X)$ generated by $X$ itself. \citet[Proposition 2]{Patton2017} also considered this notion of auto-calibration in the context of forecast dominance. In the literature on forecasting binary probabilities, which are mean forecasts and thus nested in the current setting, the same notion is often simply called `calibration', see e.g. \citet[Section 2.1]{GneitingRanjan2010}. Furthermore, 
the definition coincides with the null hypothesis of the popular \citet[henceforth MZ]{Mincer1969} regression, given by 
\begin{equation}
Y = \alpha + \beta X + \text{error};\label{mzreg}
\end{equation}
the null hypothesis $(\alpha, \beta) = (0, 1)$ corresponds to $X$ being an auto-calibrated forecast of $Y$.

Auto-calibration relates to the joint distribution of the forecast $X_j$ and the realization $Y$. Below we make use of the concept of convex order that compares univariate distributions. 

\begin{defn}[Convex order]\label{def:conv-order}
A random variable $Z_1$ is \emph{greater} than $Z_2$ in \emph{convex order} if $\E{\phi(Z_1)} \ge \E{\phi(Z_2)}$, for all convex functions $\phi$ such that the expectations exist.
\end{defn}

By \citeauthor{Strassen1965}'s (\citeyear{Strassen1965}) theorem, $Z_1$ is greater than $Z_2$ in convex order if and only if there are random variables $Z_1'$, $Z_2'$ on a joint probability space such that $Z_1' \sim Z_1$, $Z_2' \sim Z_2$ and $\EC{Z_1'}{Z_2'} = Z_2'$. Here, $\sim$ denotes equality in distribution. 
If $Z_1$ is greater than $Z_2$ in convex order then $\V{Z_1} \ge \V{Z_2}$, where $\mathbb{V}$ denotes variance. The converse is generally false; however, in the special case that $Z_1$ and $Z_2$ are both Gaussian with the same mean, $\V{Z_1} > \V{Z_2}$ implies that $Z_1$ is greater in convex order than $Z_2$.

If $Z_1$ is greater than $Z_2$ in convex order, then $-Z_2$ second-order stochastically dominates $-Z_1$. (A random variable $V$ second-order stochastically dominates another random variable $W$ if $\E{u(V)} \ge \E{u(W)}$ for all non-decreasing and concave functions $u$; see \citet[][Section 3.6]{Levy2016}. Note that this definition is weaker than convex order since the latter involves both increasing and decreasing functions $\phi$.\label{fn:ssd}) Furthermore, writing $Z_1' = Z_2' + \varepsilon$ with $\varepsilon = Z_1' - Z_2'$, we obtain $\EC{\varepsilon}{Z_2'} = 0$. In the economic literature, $Z_1$ is sometimes referred to as being equal in distribution to '$Z_2$ plus noise' \citep{RothschildStieglitz1970,PrattMachina1997}. The term `noise' for $\varepsilon$ suggests that the variation in $Z_1$ is undesirable. Indeed, if $-Z_1$ and $-Z_2$ represent two investments with stochastic monetary payoffs, then every risk-averse decision maker with concave utility function will prefer $-Z_2$ to $-Z_1$. We avoid the `noise' terminology since the negative connotation of the term is not justified in the present context; by contrast, the following result indicates that being more volatile is highly desirable in the context of auto-calibrated mean forecasts.  

\begin{thm}\label{thm:1} Assume that $A$ and $B$ are both auto-calibrated mean forecasts. Then, $A$ dominates $B$ if and only if $X_A$ is greater than $X_B$ in convex order. 
\end{thm}

According to Theorem \ref{thm:1}, it is desirable for a forecast to be large in convex order: Given the assumption that forecasts are auto-calibrated, being large in convex order implies that the forecast is more variable and is based on a `larger' information set $\mathcal{A}_j$. Without the assumption of auto-calibration, a forecast could be more variable simply because of erratic variation (see Sections \ref{sec:gauss} and \ref{sec:model} below). In the case that $Y$ is binary and $X_A, X_B$ are discretely distributed with finite support, Theorem \ref{thm:1} coincides with \citet[Theorem 1]{DegrootFienberg1983}. However, Theorem \ref{thm:1} is much more widely applicable since it imposes no assumptions on the distributions of $Y$, $X_A$ and $X_B$. 
 
\begin{exmp}\label{exmp:1} Let $Y = Z_1 + Z_2 + Z_3 + Z_4$, where $\{Z_k\}_{k=1}^4$ are independent and identically distributed random variables with mean zero. The distribution may be non-Gaussian, may involve skewness and excess kurtosis, or could be discrete. Now let $X_A = Z_1 + Z_2$ and $X_B = Z_3$, such that both $A$ and $B$ are auto-calibrated for $Y$, and $X_A$ is greater than $X_B$ in convex order. By Theorem \ref{thm:1}, $A$ dominates $B$. This setup includes the example of \citet[p.~557]{Ehm2016} where $Z_k$ are all standard normal and dominance is established via calculations that exploit normality.
\end{exmp}

\begin{exmp}\label{exmp:2} Suppose that $X_A$ and $X_B$ are both auto-calibrated and normally distributed. If $\V{X_A} > \V{X_B}$, then normality implies that $X_A$ is greater than $X_B$ in convex order, so that $A$ dominates $B$ by Theorem \ref{thm:1}. This example generalizes \citet[Proposition 2]{Patton2017} since it is based on slightly weaker assumptions and establishes dominance under all consistent scoring functions instead of a subclass called exponential Bregman loss.
\end{exmp}

\begin{prop}\label{prop:he}
For $j = A,B$, let $X_j = \EC{Y}{\mathcal{F}_j},$ where $\mathcal{F}_B \subset \mathcal{F}_A$. Then  $X_A$ and $X_B$ are both auto-calibrated and $X_A$ is greater than $X_B$ in convex order.  
\end{prop}

Examples \ref{exmp:1} and \ref{exmp:2} both feature non-nested information sets.
{Proposition \ref{prop:he} establishes that two forecasts with nested information sets satisfy the auto-calibration and convex order conditions that underlie Theorem \ref{thm:1}. The latter} then states that $X_A$ dominates $X_B$, as would be expected given that $X_A$ has access to a larger information set and both forecasts are correctly specified. The result of \citet[final line of Corollary 2]{Holzmann2014} uses the same setup as Proposition \ref{prop:he} above, and is thus a special case of Theorem \ref{thm:1}. Hence, Theorem \ref{thm:1} provides sufficient conditions for forecast dominance that are weaker than the ones by \citeauthor{Holzmann2014}. However, the result of \citeauthor{Holzmann2014} applies to general functionals, whereas we focus on the mean functional. The following example concerns forecasts made at different points in time, which is an important special case of nested information sets in practice. 

\begin{exmp}\label{exmp:3}
Let $Y_t = a~Y_{t-1} + \varepsilon_t,$ where $|a| < 1$ and $\varepsilon_t$ is independent and identically distributed with mean zero and variance $\sigma^2$, and let $\mathcal{F}_{t}$ be the information set generated by observations until time $t$. Suppose $X_{tA} = \EC{Y_t}{\mathcal{F}_{t-1}} = a~Y_{t-1}$ and $X_{tB} = \EC{Y_t}{\mathcal{F}_{t-h}} = a^h~Y_{t-h}$ for some $h \in \{2,3,\ldots\}$. Then $Y_t, X_{tA}$ and $X_{tB}$ are all strictly stationary time series, and $\mathcal{F}_{t-h} \subset \mathcal{F}_{t-1}$. Proposition \ref{prop:he} thus implies that both forecasts are auto-calibrated, and that $X_{tA}$ is greater than $X_{tB}$ in convex order. Hence, the variance of $X_{tA}$ exceeds that of $X_{tB}$, which also follows from Corollary 2 of \cite{PattonTimmermann2012}.
\end{exmp}

Finally, the following corollary describes a simple implication of Theorem \ref{thm:1} that is closely related to empirical practice in econometrics.

\begin{cor}\label{cor:mz}
Consider MZ regressions as in Equation (\ref{mzreg}), conducted separately for forecast $j \in \{A, B\}$. Suppose that $A$ and $B$ satisfy the conditions of Theorem \ref{thm:1}. Then in population, the MZ regression for $A$ attains a higher $R^2$ than the one for $B$. 
\end{cor}

This relates to the empirical literature on forecasting financial volatility, where $R^2$s of MZ regressions are commonly used to assess forecasting ability of alternative methods \citep[e.g.][Tables III.A and III.B]{AndersenEtAl2003}. See Section \ref{sec:data-vola} for an empirical illustration.

\section{Forecast Dominance under Normality}\label{sec:gauss}

Auto-calibration essentially rules out uninformative variation (`noise') in a forecast that may result from an overfitted statistical model, for example.

\begin{exmp}\label{ex:unc}
Let $Y = X_A + \varepsilon, $ where $X_A$ and $\varepsilon$ are independently standard normal. Suppose forecaster $A$ quotes $X_A$ as a mean forecast for $Y$, and forecaster $B$ quotes $X_B = X_A + \zeta,$ where $\zeta \sim \mathcal{N}(0, \sigma^2_\zeta)$, independently of $X_A$ and $\varepsilon$. One obtains easily that $\EC{Y}{X_B} = X_B/(1 + \sigma^2_\zeta)$, which implies that forecast $B$ is uncalibrated. 
\end{exmp}

In Example \ref{ex:unc}, intuition suggests that $A$ is a better forecast than $B$ since the latter simply adds the noise term $\zeta$ on top of the former. Theorem \ref{thm:1} cannot be used to derive this statement
since $B$ is uncalibrated. In this section and in Section \ref{sec:model}, we dispense with the auto-calibration assumption. In order to arrive at interpretable conditions, we investigate the scenario in which the forecast $X_j, j \in \{A, B\}$ and the realization $Y$ follow a bivariate normal distribution, such that
\begin{equation}
\begin{pmatrix} X_j \\ Y \end{pmatrix} \sim \mathcal{N}\left(\begin{pmatrix} \mu_j \\ \mu_Y \end{pmatrix}, \begin{pmatrix} \sigma^2_j & \rho_{Yj}~\sigma_j\sigma_Y \\ \rho_{Yj}~\sigma_j\sigma_Y & \sigma^2_Y \end{pmatrix}\right),\label{eq:gauss}
\end{equation}
where $\rho_{Yj} \in [-1,1]$ is the correlation between $X_j$ and $Y$. 

The Gaussian setup is similar to \cite{Satopaa2016} who motivate joint normality of forecasts and realizations from a situation in which forecasters observe small bits ('particles') of the information that generates the predictand; see their Section 3.2. 

Forecast dominance does not depend on the dependence structure between the forecasts. Hence Equation (\ref{eq:gauss}) refers to the pair $(X_j, Y)'$ only; the joint distribution of $(X_A, X_B)'$ is left unspecified, and may be non-Gaussian. The distribution in (\ref{eq:gauss}) is an unconditional one, and does not specify the dependence (or independence) across forecast instances. See {Example A.1 in the Online Appendix} for a stationary time series illustration that fits into the Gaussian framework.

We assume that $\mu_Y = \mu_A = \mu_B$, which means that forecasts $A$ and $B$ correctly assess the unconditional mean of $Y$. This simplifies our analysis but does not seem restrictive in most applications. The setup in Equation (\ref{eq:gauss}) allows for a wide range of scenarios in terms of forecast accuracy. In particular, the correlation parameter $\rho_{Yj}$ may be positive or negative, and there is no prespecified relation between the variance parameters $\sigma_j$ and $\sigma_Y$. This modeling approach hence is capable of describing the behavior of imperfect forecasts.

\begin{prop}\label{prop:bivnormal}
	Assume that for $j \in \{A, B\}$ the distribution of $(X_j, Y)$ is bivariate normal as in Equation (\ref{eq:gauss}). Then 
	\begin{multline}\label{eq:bivnormal}
	\mathbb{E}(S_\theta(X_B, Y)) - \mathbb{E}(S_\theta(X_A, Y))
	= \frac{\sigma_Y}{2}\left\{\rho_{YA}~\varphi\left(\frac{\theta-\mu_Y}{\sigma_A}\right)-
	\rho_{YB}~\varphi\left(\frac{\theta-\mu_Y}{\sigma_B}\right)\right\}   \\
+\frac{(\theta-\mu_Y)}{2}~\left\{\Phi\left(\frac{\theta-\mu_Y}{\sigma_A}\right) - \Phi\left(\frac{\theta-\mu_Y}{\sigma_B}\right)\right\}, 
	\end{multline}
	{where $S_{\theta}(x,y)$ is the elementary score function defined at (\ref{esmean}), and} $\varphi$ and $\Phi$ are the probability density and CDF of a standard normal distribution, respectively.
	\label{prop:norm}
\end{prop}

{By \citet[Theorem 1b and Corollary 1b]{Ehm2016}, $A$ dominates $B$ if the left hand side of \eqref{eq:bivnormal} is non-negative for all $\theta \in \mathbb{R}$.} The expression in \eqref{eq:bivnormal} yields several sets of sufficient conditions for forecast dominance, where we use the notation $\beta_j = \rho_{Yj}~\sigma_Y/\sigma_j$ to denote the population slope coefficient in a MZ regression of $Y$ on $X_j$ as in Equation (\ref{mzreg}). The condition $\beta_j = 1$ is necessary and sufficient for auto-calibration. 
\begin{description}
	\item[\textbf{Case 1}] Let $\sigma_A \ge \sigma_B$, and assume that {$\beta_B \le 1 \le \beta_A$.} Then $A$ dominates $B$.
	\item[\textbf{Case 2}] Let $\sigma_A \le \sigma_B$.
	\begin{description}
		\item[\textbf{Case 2a}] Assume that {$0 \le \beta_A,\beta_B \le 1$. If $\beta_A\sigma_A^2 \ge \beta_B\sigma_B^2$,} then $A$ dominates $B$.
		\item[\textbf{Case 2b}] If {$\beta_B \le 0 \le \beta_A$}, then $A$ dominates $B$.
	\end{description}
	\item[\textbf{Case 3}] Suppose that $\beta_A\sigma_A=\beta_B\sigma_B$, {{and that either $\beta_A, \beta_B > 1$ or $\beta_A, \beta_B < 1$. Then the forecast $j$ for which $|\beta_j - 1|$ is smaller dominates the other.}}
	\item[\textbf{Case 4}] If $\sigma_A = \sigma_B$, the forecast $j$ for which {$\beta_{j}$} is higher dominates the other. 
\end{description}
Justification of these claims is given in the Appendix. {For two auto-calibrated forecasts ($\beta_A = \beta_B = 1$)}, Case 1 implies that the one with higher variance is dominant, which echoes the statement of Theorem \ref{thm:1}. (Since both forecasts are Gaussian with the same mean, having higher variance is the same as being greater in convex order.) However, Case 1 does not require auto-calibration. It implies that there may be dominance relations among two uncalibrated forecasts, or dominance of an auto-calibrated forecast over an uncalibrated competitor, or vice versa. Case 2a describes a situation in which $A$ has lower variance than $B$, but at the same time has higher covariance with $Y$. This suggests that $A$ has a more favorable signal-to-noise ratio than $B$, explaining dominance of $A$ over $B$. In Case 2b, $B$ is a particularly poor forecast, featuring high variance and negative correlation with $Y$. {{Case 3 describes situations in which both forecasts have the same correlation with $Y$, and both are uncalibrated. In these situations, the forecast that comes closer to being auto-calibrated is dominant.}} Finally, Case 4 describes a simple condition for dominance if both forecasts have the same variance.

Proposition \ref{prop:norm} yields a simple necessary condition for forecast dominance: For $A$ to dominate $B$, it must hold that $\rho_{YA} \ge \rho_{YB}$. (This can be seen by evaluating the expected score difference in Proposition \ref{prop:norm} at $\theta = \mu_Y$.) If the forecast parameters satisfy this necessary condition but can not be classified into one of the four cases presented above, it is unclear whether a dominance relation exists. In this situation, one can use the result of Proposition 4.1 for an informal numerical check of dominance; see Example A.2 in the Online Appendix for an illustration.

A major implication of the Gaussian case is that auto-calibration -- which underlies Section \ref{sec:auto}, as well as all of the previous literature -- is not generally required to establish forecast dominance. In particular, there may well be dominance relations among forecasts generated from mis-specified statistical models; see Section \ref{sec:model}.

\section{Forecasts based on a Common Information Set} \label{sec:model}

{The results in Section \ref{sec:gauss} do not require auto-calibration, but require joint Gaussianity of forecasts and realizations. In this section, we present a result that requires neither auto-calibration nor Gaussianity, but assumes that both forecasts can be represented as $\EC{Y}{\mathcal{F}}$ plus noise, where the information set $\mathcal{F}$ is common across forecasting methods. The forecast methods can be viewed as different ways of exploiting $\mathcal{F}$, based on statistical models using alternative estimation algorithms or functional form assumptions, for example.

\begin{thm}\label{thm:sign-condition}
	Let $\mathcal{F}\subset \mathcal{A}$ be a $\sigma$-algebra, and let
	\begin{align*}
	Y &= \EC{Y}{\mathcal{F}} + \varepsilon, \quad X_j = \EC{Y}{\mathcal{F}} + \eta_j, \quad j\in \{A,B\},
	\end{align*}	
	where $\EC{\varepsilon}{\mathcal{F}} = 0$, and $\eta_j$ is conditionally independent of $\varepsilon$ given $\mathcal{F}$. 
	Assume that, conditionally on $\mathcal{F}$, the distributions of $\eta_A$ and $\eta_B$ are both symmetric around zero and are such that $|\eta_A|$ is smaller than $|\eta_B|$ with respect to first order stochastic dominance. Then A dominates B.
\end{thm}

Conditional independence of $\eta_j$ and $\varepsilon$ says that, given the information $\mathcal{F}$, $\eta_j$ must not contain information about $\varepsilon$. This requirement seems natural given our interpretation of $\eta_j$ {as a modeling error}. {The assumptions about $\eta_A$ and $\eta_B$ imply that the former is less variable \citep[Section 3.D]{ShakedShanthiku2007}. In the special case that $\eta_j|\mathcal{F} \sim \mathcal{N}(0,\sigma^2_j)$, the condition is satisfied if $\sigma^2_A < \sigma^2_B$. The assumption that $\EC{\eta_j}{\mathcal{F}} = 0$ implies that modeling errors are unsystematic, which seems plausible in the context of overfitted statistical models, for example. The theorem nests the case that $\eta_j = 0$ almost surely for one model $j \in \{A, B\}$. 

In contrast to Theorem \ref{thm:1} and Proposition \ref{prop:bivnormal}, the conditions of Theorem \ref{thm:sign-condition} are not directly testable for empirical data. However, we present testable implications.

\begin{prop}\label{prop:testable}
	Under the conditions of Theorem \ref{thm:sign-condition}, the following statements hold:
	\begin{itemize}
		\item[(a)] $\E{X_A} = \E{X_B} = \E{Y}$.
		\item[(b)] $\mathbb{C}ov(X_j,Y) \le \V{X_j}$ for $j \in \{A,B\},$ that is, both forecasts attain a slope coeffient $\beta_j \le 1$ in MZ regressions.
		\item[(c)] $\E{X_B^{2k}} \ge \E{X_A^{2k}}$ for all $k \in \mathbb{N}$.
	\end{itemize}
\end{prop}

Theorem \ref{thm:sign-condition} has implications for out-of-sample prediction in linear models. 

\begin{exmp}
	Let
	\[
	Y = Z'\beta + \varepsilon,
	\] 
	where $Z$ is a $p$-dimensional vector of regressors, and $\varepsilon$ is an error term satisfying $\EC{\varepsilon}{Z} = 0$. Suppose that forecast $j \in \{A,B\}$ is based on some estimator for $\beta$, obtained from training data $\{Y_i, Z_i\}_{i=1}^n$. We seek to make predictions for a new observation 
	$Y_0 = Z_0'\beta + \varepsilon_0$, where $Z_0$ and $\varepsilon_0$ are independent of the training data. We have that 
	$X_j = Z_0'\hat \beta_j^n = Z_0'\beta + Z_0'~(\hat \beta_j^n - \beta)$, where $\hat \beta_j^n$ is the estimator underlying forecast $j$, and $\eta_j=Z_0'~(\hat \beta_j^n - \beta)$ represents the approximation error of forecast $j$. Setting $\mathcal{F} = \sigma(Z_0)$, we can apply Theorem \ref{thm:sign-condition}. By assumption, $\hat \beta_j^n - \beta$ (which is generated from training data) is independent of $\varepsilon_0,$ such that $\eta_j$ is conditionally independent of $\varepsilon_0$ given $\mathcal{F}$. For large training samples, it is natural to assume multivariate normality of $\hat \beta_j^n - \beta$ for $j \in \{A,B\}$ with mean zero and covariance matrix $\Sigma_j$. Under this assumption, dominance of $A$ over $B$ occurs if $a' \Sigma_A a \le a'\Sigma_B a$, for all $a \in \mathbb{R}^k$, which is equivalent to $(\Sigma_B - \Sigma_A)$ being positive semi-definite. This is the standard notion of $A$ being a more precise estimator of $\beta$ \citep[Equation 4.4]{LehmannCasella1998}.
\end{exmp}

}

\section{Data Examples}
\label{sec:data}

\subsection{Forecasting the volatility of financial asset returns}
\label{sec:data-vola}

Following \cite{AndersenEtAl2003}, a large literature is concerned with modeling and forecasting realized measures of asset return volatility. Here we consider forecasting $\log \text{RK}_t$, where $\text{RK}_t$ is a realized kernel estimate \citep{BarndorffEtAl2008} for the Dow Jones Industrial Average on day $t$. The two forecast specifications we compare are of the form
$$\widehat{\log \text{RK}_t} = \hat \beta_0 + \hat \beta_1 Z_{t-1} + \hat \beta_2 \sum_{l = 1}^5 Z_{t-l} + \hat \beta_{3} \sum_{l=1}^{22} Z_{t-l},$$
where $\{Z_t\}_t$ is a sequence of predictor variables. This functional form follows \cite{Corsi2009}, and provides a simple way of capturing the temporal persistence in $\log \text{RK}_t$ that is typical of financial volatilities. For forecast $A$, $Z_t$ corresponds to the daily logarithmic value of the VIX index, an implied volatility index computed from financial options. For forecast $B$, $Z_t$ corresponds to the logarithmic value of the absolute index return on day $t$. We estimate both specifications using ordinary least squares, based on a rolling window of $1000$ observations. Data on the realized kernel measure and daily returns are from the Oxford-Man Realized library at \url{https://realized.oxford-man.ox.ac.uk/}; data on the VIX are from the FRED database of the Federal Reserve Bank of St. Louis (\url{https://fred.stlouisfed.org/series/VIXCLS}). The sample obtained from merging both data sources covers daily observations from January 4, 2000 to May 10, 2018. The initial part of the sample is reserved for estimating the model. We evaluate forecasts for an out-of-sample period ranging from February 13, 2004 to May 10, 2018 ($3580$ observations). 

To illustrate the conditions for Theorem \ref{thm:1} empirically, we first consider MZ regressions for both forecasts, based on the out-of-sample period. For forecast $A$ (based on VIX), we obtain the estimate
\begin{center}
	\begin{tabular}{ccccc}
		$Y_t = $ & $0.029 $ & $+$ & $1.010$ & $X_{tA}$ + error; \\
		& $[0.030]$ & 		& $[0.022]$ & \\
	\end{tabular}
\end{center}
the $R^2$ of the regression is $64 \%$, and standard errors that are robust to autocorrelation and heteroscedasticity are reported in brackets. The standard errors are computed using the function \textsf{NeweyWest} from the \textsf{R} package \textsf{sandwich} \citep{Zeileis2004}, which implements the \cite{NeweyWest1987, NeweyWest1994} variance estimator. For forecast $B$ (based on absolute returns), we obtain 
\begin{center}
	\begin{tabular}{ccccc}
		$Y_t = $ & $0.015$ & $+$ & $1.003$ & $X_{tB}$ + error, \\
		& $[0.051]$ & 		& $[0.046]$ & \\
	\end{tabular}
\end{center}
with an $R^2$ of $48.2 \%$. In both regressions, a Wald test of the hypothesis of auto-calibration (corresponding to an intercept of zero and a slope of one) cannot be rejected at conventional significance levels.

To assess the convex order condition empirically, let $F_j$ denote the CDF of forecast $j \in \{A, B\}$. Then $A$ is greater than $B$ in convex order if and only if
\begin{equation}
\int_{-\infty}^{x} F_{A}(z)~dz - \int_{-\infty}^{x} F_{B}(z)~dz \ge 0\label{eq:intco}
\end{equation}
for every $x \in \mathbb{R}$, and equality holds in the limit as $x \rightarrow \infty$ (see the proof of Theorem \ref{thm:1} in Appendix \ref{app:b}). Figure \ref{fig:cdfs_vola} plots the empirical CDFs of both forecasts. Visual inspection suggests that the integral condition in Equation \ref{eq:intco} is plausible in the current example. {In order to provide a more formal assessment, we use the subsampling based test by \cite{LintonEtAl2005} to investigate the hypothesis that one distribution is smaller than another in convex order. ({\citet{LintonEtAl2005} test for second order stochastic dominance (SOSD). Under the assumption of auto-calibration, both forecasts have the same expected value, so that SOSD and convex order coincide, except for a differential sign convention.})  We abbreviate the hypothesis of interest as `$A$ is CO-smaller than $B$' in the following discussion. Since the test depends on a tuning parameter (the size $b$ of the subsamples) that is hard to select in practice, \citet[Section 5.2]{LintonEtAl2005} suggest to plot the test's $p$-value against $b$, and select $b$ from within a range over which $p$ is stable; see Online Appendix B.1 for details. Figure \ref{fig:subsamp} shows the test results. The hypothesis that $A$ is CO-smaller than $B$ is rejected at the five percent levels for a range of $b \ge 2000$ over which the $p$-values are stable. By contrast, the right panel of Figure \ref{fig:subsamp} shows no evidence against the hypothesis that $B$ is CO-smaller than $A$, with large $p$-values for all values of $b$. In summary, the test thus reinforces the impression that a convex ordering (with $A$ being greater than $B$) is plausible in the present example.}

\begin{figure}

\begin{tabular}{cc}
 $A$ smaller than $B$ in convex order &  $B$ smaller than $A$ in convex order \\
\includegraphics[width = .5\textwidth]{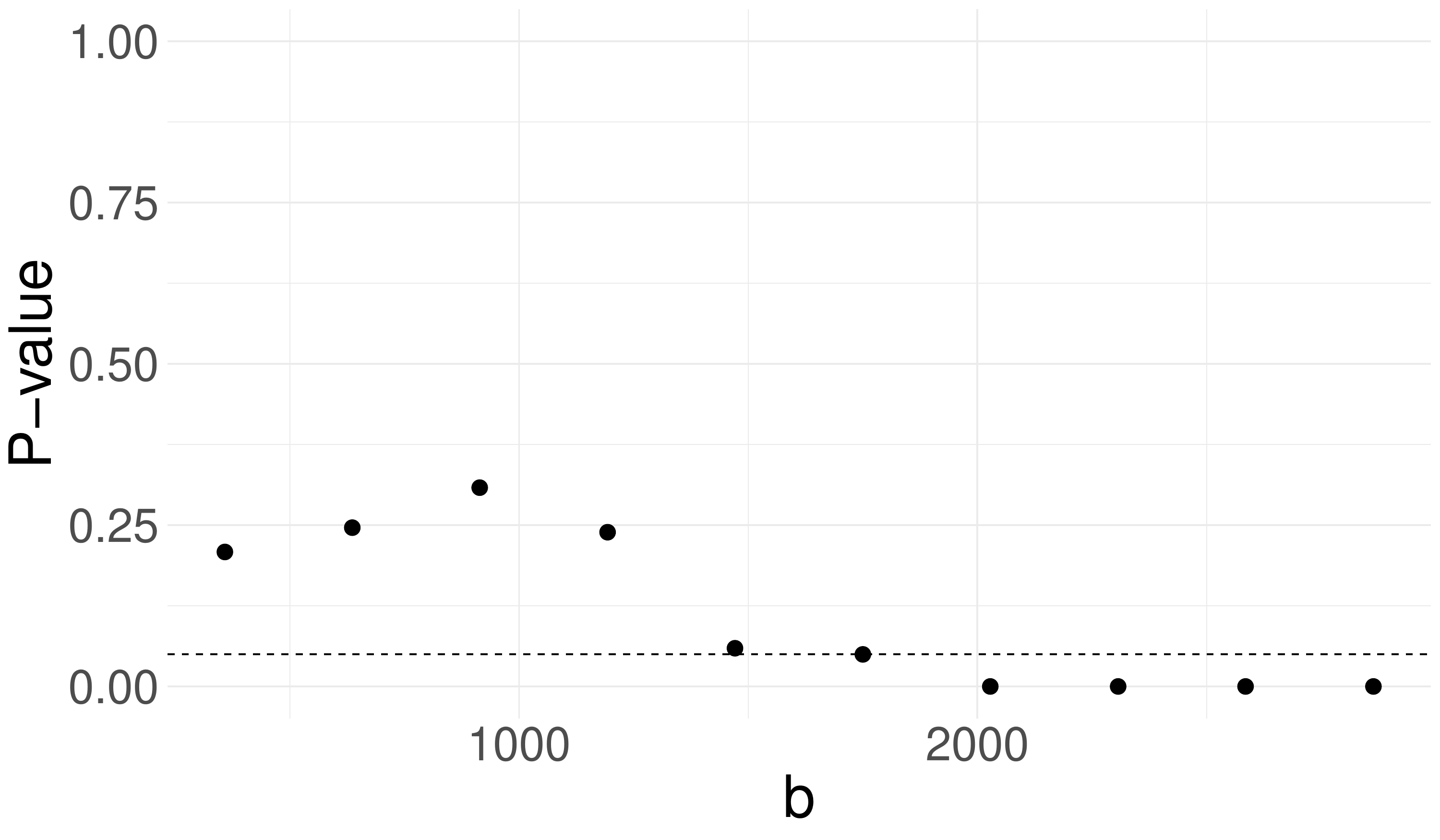} & 
\includegraphics[width = .5\textwidth]{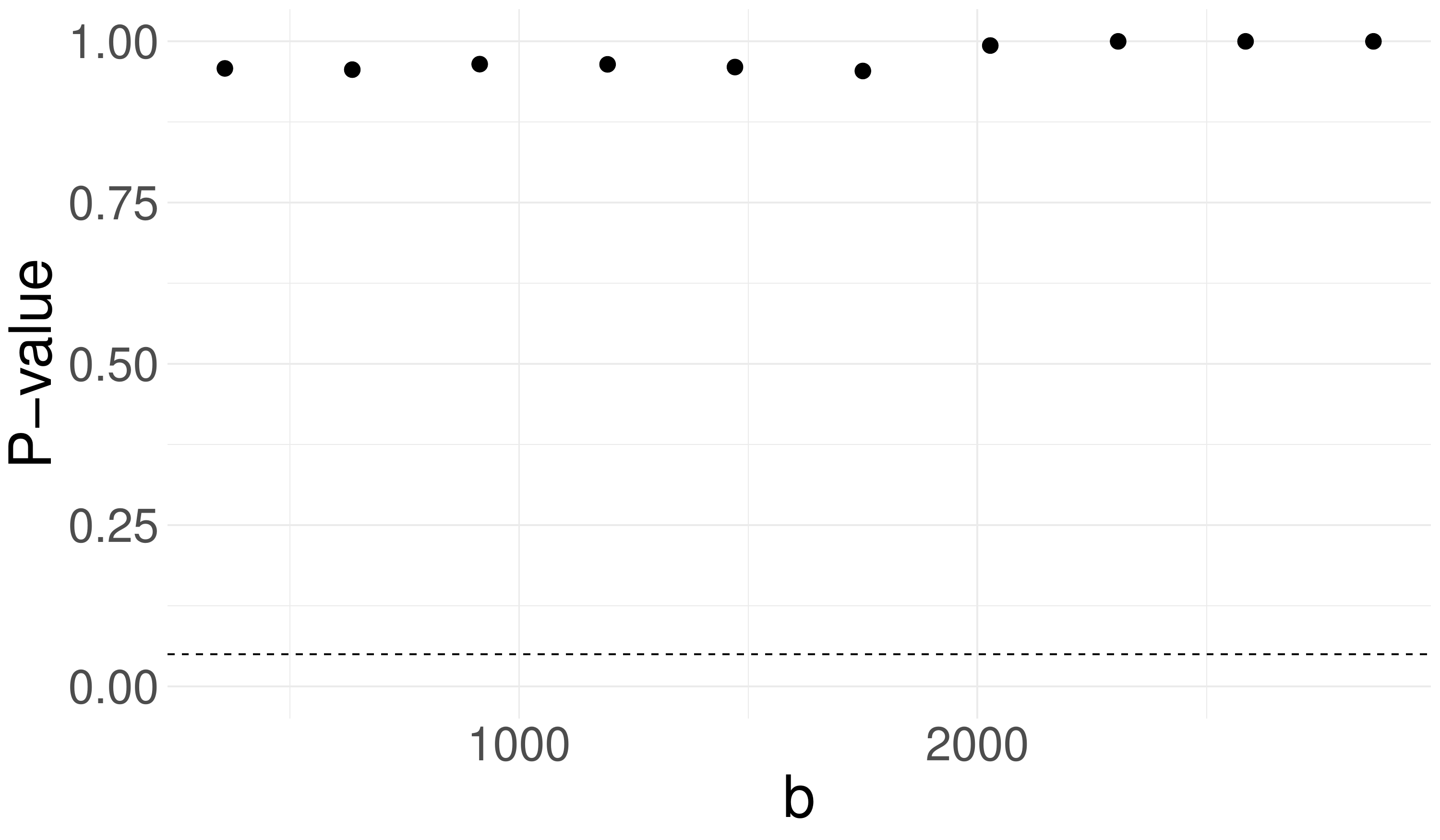} \\ 
\end{tabular}
\caption{Subsampling based $p$-values of the test by \cite{LintonEtAl2005} plotted against the subsample size parameter $b$. The dashed horizontal line marks a $p$-value of five percent. \label{fig:subsamp}}
\end{figure}

Hence both conditions of Theorem \ref{thm:1} seem plausible, and forecast $A$ appears to be more informative than forecast $B$. {Thus, we expect $A$ to dominate $B$. In order to test dominance empirically, we use the bootstrap-based test by \cite{ZiegelEtAl2017} which we modify to cover the class of Bregman scoring functions at (\ref{csf}), instead of the class of scoring functions related to Expected Shortfall that is used by \cite{ZiegelEtAl2017}. Following their implementation, we use a stationary bootstrap with block length drawn from a geometric distribution with mean $1.36~n^{-1/3},$ where $n$ is the size of the forecast evaluation sample. We use $10,000$ bootstrap iterations; see Online Appendix B.2 for further details. In line with the implication of Theorem \ref{thm:1}, the hypothesis that $A$ dominates $B$ is not rejected by the test, with a bootstrap $p$-value of one. In contrast, the hypothesis that $B$ dominates $A$ is rejected with a bootstrap $p$-value below one percent.} 

\begin{figure}[!htbp]
	\begin{center}
		\includegraphics[width = \textwidth]{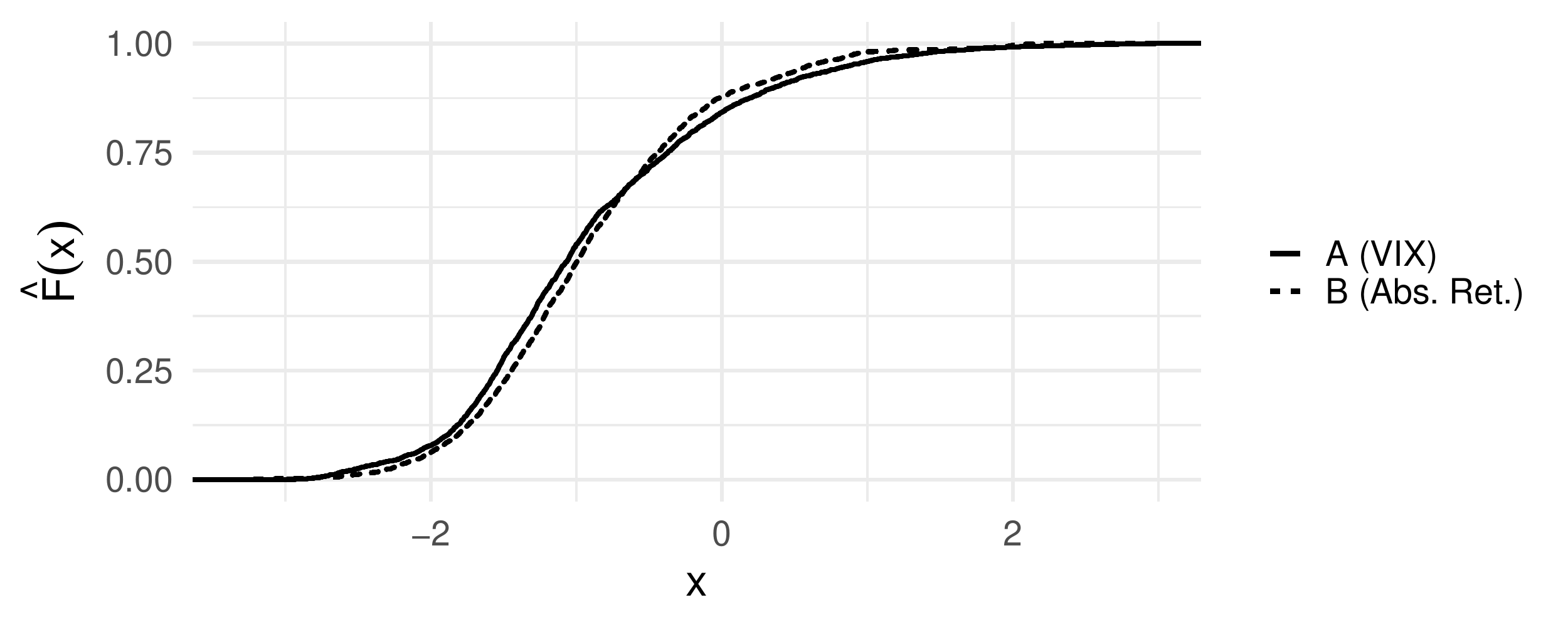}
		\caption{Volatility example: Empirical CDFs of both forecasts. \label{fig:cdfs_vola}}
	\end{center}
\end{figure}

\subsection{Forecasting US inflation}

We illustrate the results of the normally distributed case from Section \ref{sec:gauss} with inflation forecasts from the Survey of Professional Forecasters (SPF), a widely used survey of macroeconomic experts. We compare the survey against two simple forecasting schemes: A random walk forecast (RW) that states the latest realization available to SPF participants, and a rolling mean forecast (RM) considering the four latest available observations \citep{AtkesonOhanian2001}. Given their simplicity, these methods act as minimal benchmarks for more sophisticated competitors, and are routinely included in practical forecast comparisons \citep[see e.g.][Section 2.5]{FaustWright2013}. Our analysis is based on real-time data published by the Federal Reserve Bank of Philadelphia at \url{https://www.philadelphiafed.org/research-and-data/real-time-center}. We focus on inflation as measured by the GDP deflator; the relevant series codes are PGDP (SPF forecasts) and P (realizations). We compare the forecasts against the second vintages of the realizations data. We further center the forecasts and realizations at zero in order to enforce the common mean assumption made in Section \ref{sec:gauss} ($\mu_Y= \mu_A = \mu_B$); however, our results are very similar if we omit this centering step.

We first assess the assumption that forecasts $X_{tj}$ and realizations $Y_t$ follow a bivariate normal distribution. To this end, we implement the test by \cite{LobatoVelasco2004} for the null hypothesis that a univariate stationary time series is unconditionally Gaussian. The test is appealing in that it is free of tuning parameters. We apply the test to the forecasts $X_{tj}$, the outcome $Y_t$ and the forecast errors $Y_t-X_{tj}$, all of which are normally distributed if $X_{tj}$ and $Y_t$ are jointly normal. Repeating this procedure for three different forecast methods $j$ (SPF, random walk and rolling mean) and at five forecast horizons (ranging from zero to four quarters ahead), we obtain $p$-values above $20\%$ in all but one case. These results indicate that there is little evidence against pairwise bivariate normality of forecasts and realizations. Analogous tests for other macroeconomic variables (GDP growth and consumer price inflation) yielded clear rejections of normality, which is why we do not consider these variables here.

As a simple summary measure of forecast performance, Table \ref{tab:norminf_corr} presents the methods' mean squared error (MSE) at various forecast horizons. The SPF attains the smallest MSE among the three methods, with the rolling mean method performing similarly well at some horizons. The random walk method attains the largest MSE at all horizons. In order to assess the plausibility of various dominance scenarios (see below Proposition \ref{prop:norm}),  Table \ref{tab:norminf_corr} presents some relevant statistics related to the covariance matrix of $(X_{tj},Y_t)'$. We check whether these statistics match any of the scenarios under which dominance may occur. Consider, for example, the comparison of SPF versus RW at horizon $h = 0$ in the first column of Table \ref{tab:norminf_corr}. The SPF forecasts have a smaller empirical standard deviation than the random walk forecasts ($\sigma_{SPF} = 0.916 <  1.156 = \sigma_{RW}$). {At the same time, the SPF's MZ regression coefficient ($\beta_{SPF} = 0.903$) exceeds that of the random walk ($\beta_{RW} = 0.471$).} These findings indicate that the SPF forecasts have a better signal-to-noise ratio than the random walk. Indeed, the point estimates satisfy the conditions of Case 2a in Section \ref{sec:gauss}, with the SPF taking the role of the dominant forecast $A$.

The left panel of Table \ref{tab:compare} summarizes the outcomes of similar comparisons for all forecast horizons $h$. {This analysis is based on the empirical point estimates, and can hence be thought of as calibrating the theoretical results of Section \ref{sec:gauss} to empirical data.} The table reports a `$\checkmark$' entry whenever the parameters in Table \ref{tab:norminf_corr} belong to one of the sufficient conditions for dominance presented in Section \ref{sec:gauss} (Case 1-4). The SPF forecasts are dominant in six instances, all of which satisfy the conditions of Case 2a. These findings hence indicate that the SPF forecasts tend to contain less noise and more signal than the simple time series methods. Furthermore, 
according to the parameter estimates, the RM forecast dominates the RW forecast at the three shortest horizons, with the parameters again belonging to Case 2a in each case.

{The right panel of Table \ref{tab:compare} reports bootstrap $p$-values for various possible dominance relations. The bootstrap implementation is analogous to the one in Section \ref{sec:data-vola}. The bootstrap is nonparametric, contrasting the Gaussian setup of the theory in Section \ref{sec:gauss}. In comparing the left and right panels of Table \ref{tab:compare}, one can see a fairly close correspondence between the theoretical implications and the empirical test results. In particular, instances where theory predicts dominance (symbol $\checkmark$ in left panel) correspond to high bootstrap $p$-values in the right panel, such that there is no evidence against dominance. Cases where theory rules out dominance (symbol \texttt{X} in left panel) tend to go along with low bootstrap $p$-values in the right panel, corresponding to evidence against dominance.
	
The preceding analysis shows that our theoretical results under normality can inform empirical forecast comparisons. In addition, the comparisons between the two simple time series methods (RW and RM) are also in line with the theoretical conditions of Theorem \ref{thm:sign-condition}: First, both methods are based on the same information set generated by observations up until time $t$. Second, the theorem's testable implications in Proposition \ref{prop:testable} all seem plausible here; compare the coefficients $\beta_j$, $\sigma_j$ and $\mathbb{E}(X_j^4)$ reported in Table 2. The theorem then predicts dominance of RM over RW. As shown in Table 3, this conclusion is broadly in line with empirical nonparametric bootstrap tests.}

\begin{table}[!htbp]
\centering
{
\begin{tabular}{lrrrrr}
	\toprule
$h$	& $0$ & $1$ & $2$ & $3$ & $4$ \\
	\midrule
	MSE$_{SPF}$ & 0.665 & 0.778 & 0.862 & 0.920 & 1.001\\
	MSE$_{RW}$ & 1.412 & 1.535 & 1.497 & 1.348 & 1.538\\
	MSE$_{RM}$ & 0.886 & 0.917 & 0.991 & 1.103 & 1.201\\ \midrule
	$\sigma_Y$ & 1.160 & 1.160 & 1.160 & 1.160 & 1.160\\
	$\sigma_{SPF}$ & 0.916 & 0.917 & 0.967 & 1.008 & 1.012\\
	$\sigma_{RW}$ & 1.156 & 1.161 & 1.176 & 1.213 & 1.221\\
	$\sigma_{RM}$ & 0.924 & 0.935 & 0.950 & 0.971 & 0.987\\ [.2cm]
$\beta_{SPF}$ & 0.903 & 0.834 & 0.755 & 0.706 & 0.665\\
	$\beta_{RW}$ & 0.471 & 0.425 & 0.441 & 0.496 & 0.432\\
	$\beta_{RM}$ & 0.766 & 0.741 & 0.692 & 0.624 & 0.570\\ \midrule
	$\mathbb{E}(X_{SPF}^4)$ &  2.088 & 2.220 & 2.488 & 3.172 & 2.813\\
	$\mathbb{E}(X_{RW}^4)$ & 5.393 & 5.383 & 5.515 & 6.506 & 6.534\\
	$\mathbb{E}(X_{RM}^4)$ & 1.895 & 1.932 & 2.028 & 2.234 & 2.354\\
	\bottomrule
\end{tabular}
\caption{Sample estimates for the US inflation data. $h$ indicates the forecast horizon (in quarters); the sample period is 1984:Q1 to 2018:Q2. For forecast method $j \in \{SPF, RW, RM\}$, MSE$_j$ denotes the mean squared error, $\sigma_j$ denotes the standard deviation, {$\beta_{j}$ denotes the slope coefficient from a regression of realized inflation on the forecast, and $\mathbb{E}(X_j^4)$ is the fourth moment of the forecast. $\sigma_Y$ is the standard deviation of the realized inflation rates.}
\label{tab:norminf_corr}}} \end{table}

\begin{table}[!htbp]
\centering
{
\begin{tabular}{ccccccccccccc}
&& \multicolumn{5}{c}{Theory implications } &  & \multicolumn{5}{c}{Bootstrap $p$-values} \\
$h$ & & 0 & 1 & 2 & 3 & 4 & & 0 & 1 & 2 & 3 & 4 \\ \toprule
SPF $\succ_{fd}^?$ RW  & &$\checkmark$ & $\checkmark$ & $\checkmark$ & ? & $\checkmark$ & & 1.000 & 1.000 & 1.000 & 0.968 & 0.798\\
 RW $\succ_{fd}^?$ SPF && \texttt{X} & \texttt{X} & \texttt{X} & \texttt{X} & \texttt{X} & & 0.031 & 0.043 & 0.050 & 0.030 & 0.056\\
SPF $\succ_{fd}^?$ RM && $\checkmark$ & $\checkmark$ & ? & ? & ? & &0.911 & 0.650 & 0.723 & 0.441 & 0.736\\
RM $\succ_{fd}^?$ SPF && \texttt{X}& \texttt{X} & \texttt{X} & \texttt{X} & \texttt{X} & &0.435 & 0.255 & 0.160 & 0.291 & 0.225\\
RW $\succ_{fd}^?$ RM && \texttt{X}  & \texttt{X}& \texttt{X} & \texttt{X} & \texttt{X} & & 0.032 & 0.014 & 0.027 & 0.473 & 0.298\\
RM $\succ_{fd}^?$ RW && $\checkmark$  & $\checkmark$& $\checkmark$ & ? & ? & &  1.000 & 1.000 & 0.999 & 0.760 & 0.919\\\bottomrule
\end{tabular}
\caption{The notation `A $\succ_{fd}^?$ B' denotes the possibility that A dominates B. $h$ indicates the forecast horizon. Left panel: $\checkmark$ means that one of the sufficient conditions for dominance is satisfied. \texttt{X} means that the necessary condition is not satisfied. ? means that the necessary condition (but none of the sufficient conditions) is satisfied. Right panel: Bootstrap $p$-values of nonparametric forecast dominance test. \label{tab:compare}}
}
\end{table}

\section{Discussion}\label{sec:disc}

\cite{Patton2017} identifies three reasons why forecast dominance may not hold in practice: Non-nested information sets, misspecification, and estimation error. Motivated by this assessment, the present paper provides a theoretical analysis of forecast dominance that relates to each of these situations. Under the assumption that forecasts are auto-calibrated, our results in Section \ref{sec:auto} provide a novel characterization of the role played by information sets that may or may not be nested. Misspecification and estimation error are likely to lead to uncalibrated forecasts for which no analytical results are available in the existing literature on forecast dominance. Our results in Sections \ref{sec:gauss} and \ref{sec:model} cover this case in detail, based on two distinct sets of assumptions that allow us to arrive at interpretable conditions.\\ 

Conceptually, our results indicate that the notion of forecast dominance may be less strong than suggested by \cite{Patton2017}, \citet[Section 2.3]{NoldeZiegel2017}, and others. In particular, there can be dominance relations among two forecasts that are both highly imperfect. From a more technical perspective, an interesting question is whether similar conditions for forecast dominance can be derived for functionals other than the mean. As starting points of the analysis, our Theorem \ref{thm:expectiles} specifies conditions for dominance for the expectile functional (which includes the mean as a special case), and we treat quantiles in Online Appendix C. An open challenge are full distributional forecasts. 

\bibliographystyle{apalike}
\bibliography{fd}

\begin{thebibliography}{}

\bibitem[Andersen et~al., 2003]{AndersenEtAl2003}
Andersen, T.~G., Bollerslev, T., Diebold, F.~X., and Labys, P. (2003).
\newblock Modeling and forecasting realized volatility.
\newblock {\em Econometrica}, 71:579--625.

\bibitem[Atkeson and Ohanian, 2001]{AtkesonOhanian2001}
Atkeson, A. and Ohanian, L.~E. (2001).
\newblock Are {Phillips} curves useful for forecasting inflation?
\newblock {\em Federal Reserve Bank of Minneapolis Quarterly Review}, 25:2--11.

\bibitem[Barendse and Patton, 2019]{BarendsePatton2019}
Barendse, S. and Patton, A. (2019).
\newblock Comparing predictive accuracy in the presence of a loss function
  shape parameter.
\newblock Working Paper, Duke University, November 2019.

\bibitem[Barndorff-Nielsen et~al., 2008]{BarndorffEtAl2008}
Barndorff-Nielsen, O.~E., Hansen, P.~R., Lunde, A., and Shephard, N. (2008).
\newblock Designing realized kernels to measure the ex post variation of equity
  prices in the presence of noise.
\newblock {\em Econometrica}, 76:1481--1536.

\bibitem[Brier, 1950]{Brier1950}
Brier, G.~W. (1950).
\newblock Verification of forecasts expressed in terms of probability.
\newblock {\em Monthly Weather Review}, 78:1--3.

\bibitem[Buja et~al., 2005]{BujaEtAl2005}
Buja, A., Stuetzle, W., and Shen, Y. (2005).
\newblock Loss functions for binary class probability estimation and
  classification: Structure and applications.
\newblock Working Paper, University of Washington, November 2005.

\bibitem[Corsi, 2009]{Corsi2009}
Corsi, F. (2009).
\newblock A simple approximate long-memory model of realized volatility.
\newblock {\em Journal of Financial Econometrics}, 7:174--196.

\bibitem[DeGroot and Fienberg, 1983]{DegrootFienberg1983}
DeGroot, M.~H. and Fienberg, S.~E. (1983).
\newblock The comparison and evaluation of forecasters.
\newblock {\em The Statistician}, 32:12--22.

\bibitem[Ehm et~al., 2016]{Ehm2016}
Ehm, W., Gneiting, T., Jordan, A., and Kr{\"u}ger, F. (2016).
\newblock Of quantiles and expectiles: Consistent scoring functions, {C}hoquet
  representations and forecast rankings (with discussion and rejoinder).
\newblock {\em Journal of the Royal Statistical Society: Series B},
  78:505--562.

\bibitem[Ehm and Kr{\"u}ger, 2018]{EhmKruger2017}
Ehm, W. and Kr{\"u}ger, F. (2018).
\newblock Forecast dominance testing via sign randomization.
\newblock {\em Electronic Journal of Statistics}, 12:3758--3793.

\bibitem[{European Central Bank}, 2018]{ECB2018}
{European Central Bank} (2018).
\newblock {ECB} survey of professional forecasters (documentation).
\newblock Available at \url{
  https://www.ecb.europa.eu/stats/ecb_surveys/survey_of_professional_forecasters/html/index.en.html},
  accessed: September 17, 2018.

\bibitem[Faust and Wright, 2013]{FaustWright2013}
Faust, J. and Wright, J.~H. (2013).
\newblock Forecasting inflation.
\newblock In {\em Handbook of Economic Forecasting}, volume~2, pages 2--56.
  Elsevier, Amsterdam.

\bibitem[Gneiting, 2011]{Gneiting2011}
Gneiting, T. (2011).
\newblock Making and evaluating point forecasts.
\newblock {\em Journal of the American Statistical Association}, 106:746--762.

\bibitem[Holzmann and Eulert, 2014]{Holzmann2014}
Holzmann, H. and Eulert, M. (2014).
\newblock The role of the information set for forecasting--with applications to
  risk management.
\newblock {\em Annals of Applied Statistics}, 8:595--621.

\bibitem[Krzysztofowicz and Long, 1990]{KrzysztofowiczLong1990}
Krzysztofowicz, R. and Long, D. (1990).
\newblock Fusion of detection probabilities and comparison of multisensor
  systems.
\newblock {\em IEEE Transactions on Systems, Man, and Cybernetics},
  20:665--677.

\bibitem[Lehmann and Casella, 1998]{LehmannCasella1998}
Lehmann, E.~L. and Casella, G. (1998).
\newblock {\em Theory of Point Estimation}.
\newblock Springer, New York, 2 edition.

\bibitem[Levy, 2016]{Levy2016}
Levy, H. (2016).
\newblock {\em Stochastic Dominance: Investment Decision Making Under
  Uncertainty}.
\newblock Springer, New York, 3 edition.

\bibitem[Linton et~al., 2005]{LintonEtAl2005}
Linton, O., Maasoumi, E., and Whang, Y.-J. (2005).
\newblock Consistent testing for stochastic dominance under general sampling
  schemes.
\newblock {\em Review of Economic Studies}, 72:735--765.

\bibitem[Lobato and Velasco, 2004]{LobatoVelasco2004}
Lobato, I.~N. and Velasco, C. (2004).
\newblock A simple test of normality for time series.
\newblock {\em Econometric Theory}, 20:671--689.

\bibitem[Machina and Pratt, 1997]{PrattMachina1997}
Machina, M. and Pratt, J. (1997).
\newblock Increasing risk: Some direct constructions.
\newblock {\em Journal of Risk and Uncertainty}, 14:103--127.

\bibitem[Mincer and Zarnowitz, 1969]{Mincer1969}
Mincer, J.~A. and Zarnowitz, V. (1969).
\newblock The evaluation of economic forecasts.
\newblock In Mincer, J.~A., editor, {\em Economic Forecasts and Expectations:
  Analysis of Forecasting Behavior and Performance}, pages 3--46. Columbia
  University Press, New York.

\bibitem[M\"{u}ller and R\"uschendorf, 2001]{MullerRuschendo2001}
M\"{u}ller, A. and R\"uschendorf, L. (2001).
\newblock On the optimal stopping values induced by general dependence
  structures.
\newblock {\em Journal of Applied Probability}, 38:672--684.

\bibitem[Newey and Powell, 1987]{NeweyPowell1987}
Newey, W.~K. and Powell, J.~L. (1987).
\newblock Asymmetric least squares estimation and testing.
\newblock {\em Econometrica}, 55:819--847.

\bibitem[Newey and West, 1987]{NeweyWest1987}
Newey, W.~K. and West, K.~D. (1987).
\newblock A simple, positive semi-definite, heteroscedasticity and
  autocorrelation consistent covariance matrix.
\newblock {\em Econometrica}, 55:703--708.

\bibitem[Newey and West, 1994]{NeweyWest1994}
Newey, W.~K. and West, K.~D. (1994).
\newblock Automatic lag selection in covariance matrix estimation.
\newblock {\em Review of Economic Studies}, 61:631--653.

\bibitem[Nolde and Ziegel, 2017]{NoldeZiegel2017}
Nolde, N. and Ziegel, J.~F. (2017).
\newblock Elicitability and backtesting: {P}erspectives for banking regulation
  (with discussion and rejoinder).
\newblock {\em Annals of Applied Statistics}, 11:1833--1874.

\bibitem[Patton, 2011]{Patton2011}
Patton, A.~J. (2011).
\newblock Volatility forecast comparison using imperfect volatility proxies.
\newblock {\em Journal of Econometrics}, 160:246--256.

\bibitem[Patton, 2018]{Patton2017}
Patton, A.~J. (2018).
\newblock Comparing possibly misspecified forecasts.
\newblock {\em Journal of Business \& Economic Statistics}.
\newblock Forthcoming.

\bibitem[Patton and Timmermann, 2012]{PattonTimmermann2012}
Patton, A.~J. and Timmermann, A. (2012).
\newblock Forecast rationality tests based on multi-horizon bounds.
\newblock {\em Journal of Business \& Economic Statistics}, 30:1--17.

\bibitem[Ranjan and Gneiting, 2010]{GneitingRanjan2010}
Ranjan, R. and Gneiting, T. (2010).
\newblock Combining probability forecasts.
\newblock {\em Journal of the Royal Statistical Society. Series B}, 72:71--91.

\bibitem[Rothschild and Stiglitz, 1970]{RothschildStieglitz1970}
Rothschild, M. and Stiglitz, J.~E. (1970).
\newblock Increasing risk: I. {A} definition.
\newblock {\em Journal of Economic Theory}, 2:225 -- 243.

\bibitem[Satop{\"a}{\"a} et~al., 2016]{Satopaa2016}
Satop{\"a}{\"a}, V.~A., Pemantle, R., and Ungar, L.~H. (2016).
\newblock Modeling probability forecasts via information diversity.
\newblock {\em Journal of the American Statistical Association},
  111:1623--1633.

\bibitem[Savage, 1971]{Savage1971}
Savage, L.~J. (1971).
\newblock Elicitation of personal probabilities and expectations.
\newblock {\em Journal of the American Statistical Association}, 66:783--801.

\bibitem[Shaked and Shanthikumar, 2007]{ShakedShanthiku2007}
Shaked, M. and Shanthikumar, J.~G. (2007).
\newblock {\em Stochastic Orders}.
\newblock Springer, New York.

\bibitem[Str{\"a}hl and Ziegel, 2017]{StraehlZiegel2017}
Str{\"a}hl, C. and Ziegel, J.~F. (2017).
\newblock Cross-calibration of probabilistic forecasts.
\newblock {\em Electronic Journal of Statistics}, 11:608--639.

\bibitem[Strassen, 1965]{Strassen1965}
Strassen, V. (1965).
\newblock The existence of probability measures with given marginals.
\newblock {\em Annals of Mathematical Statistics}, 36:423--439.

\bibitem[Yen and Yen, 2018]{YenYen2018}
Yen, T.-J. and Yen, Y.-M. (2018).
\newblock Testing forecast accuracy of expectiles and quantiles with the
  extremal consistent loss functions.
\newblock Working Paper, National Chengchi University, July 2018.

\bibitem[Zeileis, 2004]{Zeileis2004}
Zeileis, A. (2004).
\newblock Econometric computing with {HC} and {HAC} covariance matrix
  estimators.
\newblock {\em Journal of Statistical Software}, 11:1--17.

\bibitem[Ziegel et~al., 2018]{ZiegelEtAl2017}
Ziegel, J.~F., Kr\"uger, F., Jordan, A., and Fasciati, F. (2018).
\newblock Robust forecast evaluation of {E}xpected {S}hortfall.
\newblock {\em Journal of Financial Econometrics}.
\newblock Forthcoming.

\end{thebibliography}

\appendix

\section*{Appendix}

\section{Result for Dominance of Expectile Forecasts} \label{app:A}

We state and prove a more general version of Theorem \ref{thm:mean}. We consider the expectile functional of $Y$ at level $\tau \in (0,1)$ \citep{NeweyPowell1987}. The expectile is the unique value $t$ that satisfies
$$(1-\tau) \int_{(-\infty,t]} (t-y)\diff F(y) = \tau \int_{[t,\infty)}(y-t) \diff F(y),$$
where $F(y)$ is the CDF of $Y$. The mean functional is obtained as a special case for $\tau = 1/2$. {As shown by \cite{Gneiting2011}, the class of consistent scoring functions for the expectile at level $\tau$ is given by 
\begin{equation}
S(x,y) = |\ind_{(y < x)} - \tau|~\left(\phi(y) - \phi(x) - \phi'(x)~(y-x)\right),\label{csf_exp}
\end{equation}
where $\phi$ is a convex function with subgradient $\phi'$. The relevant class for the mean (see Equation \ref{csf}) emerges for $\tau = 1/2$. Analogously to Definition \ref{def:dom}, we then have the following definition of forecast dominance for expectiles.

\begin{defn}[Forecast dominance for expectiles]\label{def:dom_exp}
Forecast $A$ \emph{dominates} forecast $B$ if $$\E{S(X_A, Y)} \le \E{S(X_B, Y)}$$
for every function $S$ of the form given in (\ref{csf_exp}).
\end{defn}}

\begin{lem}\label{lem:A.1}
For any Borel set $A \subset \mathbb{R}$,
\begin{align*}
\E{(X-Y)_+\mathbf{1}_A(X)} & = \int_{-\infty}^{\infty} \prob(Y < w, X > w, X \in A)\diff w,\\
\E{(Y-X)_+\mathbf{1}_A(X)}  &= \int_{-\infty}^{\infty} \prob(Y \ge  w, X \le w, X \in A)\diff w.
\end{align*}
\end{lem}
{
\begin{proof}
By Fubini's theorem, we obtain
\begin{align*}
\E{(X-Y)_+\mathbf{1}_A(X)} & = \int_{\mathbb{R}}\int_{\mathbb{R}}(x-y)_+ \mathbf{1}_A(x) \diff F(y|X=x) \diff G(x)\\
&= \int_{\mathbb{R}}\mathbf{1}_A(x)  \int_{(-\infty,x]}(x-y)\diff F(y|X=x) \diff G(x)\\
&= \int_{\mathbb{R}}\mathbf{1}_A(x)\int_{(-\infty,x]}\int_y^x \diff w \diff F(y|X=x) \diff G(x)\\
&= \int_{\mathbb{R}}\int_{\mathbb{R}}\int_{-\infty}^{\infty}\mathbf{1}_{(-\infty, w)}(y)\mathbf{1}_{(w, \infty)}(x)\mathbf{1}_A(x) \diff w \diff F(y|X=x) \diff G(x)\\
&= \int_{-\infty}^\infty \int_{(w,\infty)}\mathbf{1}_A(x) \int_{(-\infty,w)} \diff F(y|X=x)   \diff G(x)\diff w \\
&= \int_{-\infty}^\infty \E{\mathbf{1}_{(w,\infty)}(X)\mathbf{1}_{(-\infty,w)}(Y)\mathbf{1}_A(X)} \diff w 
\end{align*}
where $F(\cdot|X=x)$ denotes the conditional CDF of $Y$ given $X=x$, and $G$ denotes the CDF of $X$. The proof of the second equality is analogous.
\end{proof}}

\begin{lem}\label{lem:A.2}
Let $X,Z$ be two random variables such that $\E{XZ}$ exists and is finite.
Then,
\begin{align}
\E{XZ} &= \int_0^\infty\int_0^\infty \big(H(x,z) - F(x) - G(z) + 1 \big)\diff x \diff z + \int_{-\infty}^0\int_{-\infty}^0 H(x,z) \diff x \diff z\nonumber\\
&\quad +\int_{-\infty}^0\int_{0}^\infty \big(H(x,z) - G(z)\big)\diff x \diff z  +\int_0^\infty\int_{-\infty}^0 \big(H(x,z) - F(x)\big)\diff x \diff z,\label{eq:lemA.2}
\end{align}
where $H(x,z) = \prob(X \le x,Z\le z)$, $F(x) = \prob(X \le x)$, $G(z) = \prob(Z\le z)$ are the joint and marginal CDFs of $(X,Z)$, $X$ and $Z$, respectively. 
\end{lem}
{{
\begin{proof}
For a random variable $Y$, we can write
\begin{align*}
Y_+ &=\int_0^\infty (1 - \mathbf{1}_{[Y,\infty)}(x)) \diff x,\quad Y_- = \int_{-\infty}^0 \mathbf{1}_{[Y,\infty)}(x) \diff x,
\end{align*}
where $Y_+ = \max\{Y,0\}$, $Y_-=\max\{-Y,0\}$ are the positive and the negative part of $Y$, respectively. Therefore,
\begin{multline}\label{eq:XZ}
(XZ)_+ = X_+ Z_+ + X_-Z_- = \int_0^\infty\int_0^\infty (1- \mathbf{1}_{[X,\infty)}(x))(1 - \mathbf{1}_{[Z,\infty)}(z)) \diff x \diff z \\+ \int_{-\infty}^0\int_{-\infty}^0 \mathbf{1}_{[X,\infty)}(x)\mathbf{1}_{[Z,\infty)}(z) \diff x \diff z.
\end{multline}
Taking the expectation in \eqref{eq:XZ} and using Fubini's theorem, we obtain
\begin{equation*}
\E{(XZ)_+} = \int_0^\infty\int_0^\infty H(x,z) - F(x) - G(z) + 1 \diff x \diff z + \int_{-\infty}^0\int_{-\infty}^0 H(x,z) \diff x \diff z,
\end{equation*}
and, similarly, with $(XZ)_- = X_+ Z_- + X_-Z_+$,
\begin{equation*}
\E{(XZ)_-} =  \int_{-\infty}^0\int_{0}^\infty G(z) - H(x,z)\diff x \diff z  +\int_0^\infty\int_{-\infty}^0 F(x) - H(x,z)\diff x \diff z. \qedhere
\end{equation*}
\end{proof}}}
{\begin{lem}\label{lem:A.3}
The elementary scoring function for expectiles in \citet[Equation 12]{Ehm2016} is identical to the function
\begin{equation}\label{eq:lemA.3}
S_\theta(x,y) = |\mathbf{1}_{(y < \theta)} - \tau|(\theta - y)\mathbf{1}_{(x > \theta)},
\end{equation}
up to a difference of $\tau~(y-\theta)_+$ which does not depend on $x$.
\end{lem}
\begin{proof}
Adjusting the notation in Equation (12) of Ehm et al. (2016) (using the symbol $\tau$ instead of $\alpha$ for the expectile level), we have
$$S_{\theta}(x,y) = |\ind_{(y<x)}-\tau|~\left\{(y-\theta)_+-(x-\theta)_+-(y-x)~\ind_{(\theta<x)}\right\}.$$
Since $|\ind_{(y<x)}-\tau| = \ind_{(y<x)}(1-2\tau)+\tau$ and $(z)_+ = z\ind_{(z>0)}$, the score can be rewritten as 
$$S_{\theta}(x,y) =\left[\ind_{(y<x)}(1-2\tau)+\tau\right](\theta-y)\left[\ind_{(x > \theta)} - \ind_{(y>\theta)}\right].$$
Subtracting the term $\tau(y-\theta)\ind_{(y>\theta)}$ (which does not depend on $x$) and rearranging, we get
\begin{align*}
S_{\theta}(x,y) &=\left[\ind_{(y<x)}(1-2\tau)\right](y-\theta)\left(\ind_{(y > \theta)} - \ind_{(x > \theta)}\right) + \tau(\theta-y)\ind_{(x > \theta)} \nonumber\\
&= \begin{cases} \left[\ind_{(y<x)}(1-2\tau)+\tau\right](\theta-y)\ind_{(x > \theta)} & y \le \theta\\
\left[\ind_{(y<x)}(1-2\tau)\right](y-\theta)\left(1 - \ind_{(x > \theta)}\right) + \tau(\theta-y)\ind_{(x > \theta)} & y > \theta \end{cases}\\
&= \begin{cases} (1-\tau)(\theta-y)\ind_{(x > \theta)} & y \le \theta\\
 \tau(\theta-y)\ind_{(x > \theta)} & y > \theta \end{cases} \\ 
 &= |\ind_{(y \le \theta)} - \tau| (\theta-y)\ind_{(x > \theta)}= |\ind_{(y < \theta)} - \tau| (\theta-y)\ind_{(x > \theta)}.\qedhere
\end{align*}	
\end{proof}}

\begin{thm}\label{thm:expectiles}
Let $A$ and $B$ be forecasts for the $\tau$-expectile. Then $A$ dominates $B$ if and only if $\psi_{A}(\theta) \ge \psi_B(\theta)$, for all $\theta \in \mathbb{R}$,
where 
\begin{multline*}
\psi_j(\theta) =\int_{\theta}^\infty \tau \prob(X_j > w, Y > w) + (1-\tau) \prob(X_j > w ,Y \le w)\diff w\\
 + \tau \E{(Y-X_j)_+\mathbf{1}_{(X_j > \theta)}} - (1-\tau) \E{(X_j-Y)_+ \mathbf{1}_{(X_j > \theta)}}, \quad \text{for $j \in \{A,B\}$.}
\end{multline*}
\end{thm}
\begin{proof}
By \citet[Corollary 1b]{Ehm2016}, A dominates B if and only if $\E{S_\theta(X_B, Y)} \ge \E{S_\theta(X_A, Y)}$ for all $\theta \in \mathbb{R}$,
where $S_\theta$ is given at \eqref{eq:lemA.3}, see Lemma \ref{lem:A.3}. Note that $S_\theta(X_j,Y)$ is integrable if $Y$ is integrable, $j \in \{A,B\}$.
We apply Lemma \ref{lem:A.2} to the random variables $\mathbf{1}_{(X_j > \theta)}$ and $|\mathbf{1}_{(Y < \theta)} - \tau|(\theta - Y)$. We have $F(x) = \prob(\mathbf{1}_{(X_j > \theta)} \le x)$ $= \mathbf{1}_{(x \ge 1)} + \mathbf{1}_{(x \in [0,1))}\prob(X_j \le \theta)$,
\begin{align*}
G(z) &= \prob(|\mathbf{1}_{(Y < \theta)} - \tau|(\theta - Y) \le z)\\ 
&= \prob((1-\tau)(\theta-Y) \le z, Y < \theta) + \prob(\tau(\theta-Y) \le z, Y \ge \theta)\\
&= \mathbf{1}_{(z > 0)}\prob(Y \ge \theta - z/(1-\tau)) +  \mathbf{1}_{(z \le 0)}\prob(Y \ge \theta - z/\tau),\\
H(x,z) &= \prob(\mathbf{1}_{(X_j > \theta)} \le x,|\mathbf{1}_{(Y < \theta)} - \tau|(\theta - Y) \le z)\\&= \mathbf{1}_{(x \ge 1)}G(z) + \mathbf{1}_{(x \in [0,1), z > 0)}\prob(X_j \le \theta, Y \ge \theta - z/(1-\tau))\\ &\qquad+  \mathbf{1}_{(x \in [0,1), z \le 0)}\prob(X_j \le \theta, Y \ge \theta - z/\tau).
\end{align*}
Therefore, the first integral on the right hand side of \eqref{eq:lemA.2} is
\begin{align*}
&\int_0^{\infty}\int_0^\infty H(x,z) - F(x) - G(z) + 1 \diff x \diff z \\&= \int_0^\infty \int_0^1 \prob(X_j \le \theta, Y \ge \theta - z/(1-\tau)) - \prob(X_j \le \theta)) - \prob(Y \ge \theta - z/(1-\tau)) + 1 \diff x \diff z\\
&= \int_0^\infty \prob(X_j > \theta, Y < \theta - z/(1-\tau))\diff z.
\end{align*}
Similarly, we can compute the third integal on the right hand side of \eqref{eq:lemA.2} to obtain
\begin{align*}
\int_{-\infty}^0\int_0^\infty H(x,z) - G(z) \diff x \diff z = -\int_{-\infty}^0 \prob(X_j > \theta, Y \ge \theta - z/\tau)\diff z.
\end{align*}
The second and the fourth integral on the right hand side of \eqref{eq:lemA.2} are zero because $H(x,z)$ and $F(x)$ are zero for $x < 0$. 
Using a change of variables, we obtain
\begin{align*}
\psi_j(\theta)& = -\E{S_\theta(X_j,Y)} \\&=  \tau\int_{\theta}^\infty \prob(X_j > \theta, Y \ge w)\diff w - (1-\tau) \int_{-\infty}^\theta \prob(X_j > \theta, Y < w)\diff w.
\end{align*}
We can rewrite this as
\begin{align*}
\psi_j(\theta)& = \int_{\theta}^\infty \tau \prob(X_j > w, Y \ge w) + (1-\tau) \prob(X_j > w ,Y < w)\diff w\\
&\quad+ \tau\int_{\theta}^\infty\prob(w\ge X_j > \theta ,Y \ge w)\diff w \\&\quad- (1-\tau)\left(\int_{\theta}^{\infty} \prob(X_j > w,Y <  w)\diff w + \int_{-\infty}^\theta \prob(X_j > \theta, Y < w)\diff w\right)\\
&=  \int_{\theta}^\infty \tau \prob(X_j > w, Y \ge w) + (1-\tau) \prob(X_j > w ,Y < w)\diff w\\
&\quad+ \tau\int_{-\infty}^\infty\prob(X_j \le w ,Y \ge w,X_j > \theta)\diff w\\
&\quad- (1-\tau)\int_{-\infty}^{\infty} \prob(X_j > w,Y <  w, X_j > \theta)\diff w\\
&=  \int_{\theta}^\infty \tau \prob(X_j > w, Y \ge w) + (1-\tau) \prob(X_j > w ,Y < w)\diff w\\
&\quad + \tau \E{(Y-X_j)_+\mathbf{1}_{(X_j > \theta)}} - (1-\tau) \E{(X_j-Y)_+ \mathbf{1}_{(X_j > \theta)}},
\end{align*}
where the second equality holds because $\prob(w\ge X_j > \theta ,Y \ge w) = 0$ for $w < \theta$ and $\prob(X_j > w,Y <  w, X_j > \theta) = \prob(X_j > w,Y <  w)$ for $w \ge \theta$, $\prob(X_j > w,Y <  w, X_j > \theta) = \prob(Y <  w, X_j > \theta)$ for $w < \theta$. The last equality follows from Lemma \ref{lem:A.1} with $A = (\theta,\infty)$.
\end{proof}

\section{Proofs and Technical Details}
\label{app:b}

\begin{proof}[Proof of Theorem \ref{thm:mean}]
The result follows from Theorem \ref{thm:expectiles} with $\tau=1/2$ because $\prob(X_j > w, Y \ge w) +  \prob(X_j > w ,Y < w) = \prob(X_j > w)$ and $\E{((Y-X_j)_+ - (X_j-Y)_+)\mathbf{1}_{(X_j > \theta)}}$ $=\E{(Y-X_j)\mathbf{1}_{(X_j > \theta)}}$ $=\mathbb{E}((\EC{Y}{X_j}-X_j)\mathbf{1}_{(X_j > \theta)})$, where the second equality uses the law of iterated expectations.
\end{proof}

\begin{proof}[Proof of Theorem \ref{thm:1}]

Under auto-calibration, $\EC{Y}{X_j} = X_j$ holds almost surely. In view of Theorem \ref{thm:mean}, Theorem \ref{thm:1} then follows from \citet[Corollary 4.1]{MullerRuschendo2001} which shows that {$X_A$} is greater than {$X_B$} in convex order if and only if $\int_a^{\infty} \prob(X_A > t) \diff t \ge \int_a^{\infty} \prob(X_B > t) \diff t$ for all $a \in \mathbb{R}$, and 
\begin{align}
&\lim_{a \to -\infty} \left(\int_a^{\infty} \prob(X_A > t) \diff t -  \int_a^{\infty} \prob(X_B > t) \diff t\right) = 0.\label{Cor4.1}
\end{align}
{To see why \eqref{Cor4.1} holds, note that
\begin{align*}
\lim_{a \to -\infty} \left(\int_a^{\infty} \prob(X_A > t) \diff t -  \int_a^{\infty} \prob(X_B > t) \diff t\right) &= \E{X_A}  - \E{X_B} = \E{Y}-\E{Y} = 0,
\end{align*}
where the first equality follows from \citet[Proposition 4.1.(a)(iii)]{MullerRuschendo2001}, and the second equality follows from auto-calibration.}
\end{proof}

\begin{proof}[Proof of Proposition \ref{prop:he}]

Auto-calibration of $X_j$ holds because $\sigma(X_j) \subseteq \mathcal{F}_j$ and $\EC{Y}{X_j} = \EC{\EC{Y}{\mathcal{F}_j}}{X_j} = \EC{X_j}{X_j} = X_j,$ where the first equality uses the tower property of conditional expectation. To show that $X_A$ is greater than $X_B$ in convex order, note that $\EC{X_A}{X_B} = \EC{\EC{Y}{\mathcal{F}_A}}{X_B} = \EC{Y}{X_B} = X_B,$ where the second equality again uses the tower property, together with the fact that $\sigma(X_B) \subset \mathcal{F}_A$. \citeauthor{Strassen1965}'s \citeyear{Strassen1965} characterization mentioned in Section \ref{sec:mean} thus implies that $X_A$ is greater than $X_B$ in convex order.\end{proof}

\begin{proof}[Proof of Corollary at the end of Section \ref{sec:auto}]

Due to auto-calibration, $\text{Cov}(X_j, Y) = \V{X_j}$ for $j \in \{A, B\},$ where $\text{Cov}$ denotes covariance. The convex order condition implies that $ \V{X_A} \ge  \V{X_B}$, and hence that $\text{Cor}(X_A, Y) = \sqrt{R^2_A} \ge \text{Cor}(X_B, Y) = \sqrt{R^2_B},$ where $\text{Cor}$ denotes correlation and $R^2_j$ is the $R^2$ from the Mincer-Zarnowitz regression for forecast $j$.
\end{proof}

\begin{proof}[Proof of Proposition \ref{prop:norm}]

Suppose that $(X_j,Y)$ follow a bivariate normal distribution. We compute $\psi_j(\theta)$ defined in Theorem \ref{thm:mean} for $j \in \{A,B\}$. We have that $\EC{Y}{X_j}$ $= \mu_Y + \rho_{Yj}(\sigma_Y/\sigma_j)(X_j - \mu_j)$,
and hence
\begin{align*}
&\E{(\EC{Y}{X_j} - X_j)\mathbf{1}_{(X_j > \theta)}} = \E{\left(\mu_Y + \rho_{Yj}\frac{\sigma_Y}{\sigma_j}(X_j - \mu_j)-X_j\right)\mathbf{1}_{(X_j > \theta)}}\\
&\quad= \left(\mu_Y -\theta - \rho_{Yj}\frac{\sigma_Y}{\sigma_j}(\mu_j - \theta)\right)\left(1 - \Phi\left(\frac{\theta - \mu_j}{\sigma_j}\right)\right) + \left(\rho_{Yj}\frac{\sigma_Y}{\sigma_j} - 1\right)\sigma_j\Psi\left(\frac{\theta-\mu_j}{\sigma_j}\right),
\end{align*}
where we define for $\theta \in \mathbb{R}$, $\Psi(\theta) = \int_{\theta}^\infty 1 - \Phi(w)\diff w$. Then,
\begin{align}
\psi_j(\theta) &= \frac{\sigma_j}{2}\Psi\left(\frac{\theta-\mu_j}{\sigma_j}\right) + \frac{1}{2}\E{(\EC{Y}{X_j} - X_j)\mathbf{1}_{(X_j > \theta)}} \nonumber\\
&= \frac{1}{2}\left(\mu_Y -\theta - \rho_{Yj}\frac{\sigma_Y}{\sigma_j}(\mu_j - \theta)\right)\left(1 - \Phi\left(\frac{\theta - \mu_j}{\sigma_j}\right)\right) + \frac{\rho_{Yj}\sigma_Y}{2}\Psi\left(\frac{\theta-\mu_j}{\sigma_j}\right). \label{psi1}
\end{align}

Using the assumption that $\mu_A = \mu_B = \mu_Y$ and the fact that $\Psi(\theta) = \varphi(\theta) - \theta~(1-\Phi(\theta)),$ Equation (\ref{psi1}) yields that 
\begin{equation}
2~\psi_j(\theta) = \rho_{Yj}\sigma_Y~\varphi\left(\frac{\theta-\mu_Y}{\sigma_j}\right) - (\theta-\mu_Y)~\left(1-\Phi\left(\frac{\theta-\mu_Y}{\sigma_j}\right)\right).\label{psi2}\qedhere
\end{equation}
\end{proof}

\begin{proof}[Notes on Cases 1 to 4]
Case 1 holds because, for each $\theta \in \mathbb{R}$, we have $2\psi_A(\theta) + (\theta-\mu_Y) \ge \sigma_A~\varphi((\theta - \mu_Y)/\sigma_A) + (\theta - \mu_Y)\Phi((\theta - \mu_Y)/\sigma_A) \ge \sigma_B\varphi((\theta - \mu_Y)/\sigma_B) + (\theta - \mu_Y)\Phi((\theta - \mu_Y)/\sigma_B) \ge 2\psi_B(\theta) + (\theta-\mu_Y),$ where $\psi_j(\theta)$ has been defined at (\ref{psi2}). Case 2a can be shown by re-parametrizing $\sigma_{Yj} = \sigma_Y\sigma_j\rho_{Yj}$, and differentiating $2\psi_j(\theta)$ with respect to $\sigma_j$. Case 3 can be shown by differentiating $2\psi_j(\theta)$ with respect to $\sigma_j$. Cases 2b and 4 are immediate.
\end{proof}

\begin{proof}[Proof of Theorem \ref{thm:sign-condition}]	

Denote the CDF of $\eta_j$, conditional on $\mathcal{F}$, by $F_j^{\mathcal{F}}$, for $j \in \{A, B\}$. By \citet[Theorem 3.D.1]{ShakedShanthiku2007}, the assumptions of Theorem \ref{thm:sign-condition} imply that 
	\begin{equation}\label{eq:sign-change-condition}
	F_A^{\mathcal{F}}(z) - F_B^{\mathcal{F}}(z) \begin{cases}\ge 0, & \text{for $z \ge 0$,}\\ \le 0, & \text{for $z \le 0$.}\end{cases}
	\end{equation}
	By \citet[Corollary 1b]{Ehm2016}, A dominates B if and only if $\E{S_\theta(X_B, Y)} \ge \E{S_\theta(X_A, Y)}$ for all $\theta \in \mathbb{R}$,
	where $S_\theta$ is given at \eqref{esmean}. The random variable $S_\theta(X_j,Y)$ is integrable if $Y$ is integrable. Define $W = \EC{Y}{\mathcal{F}}$, and let $\theta \in \mathbb{R}$. Then,
	\begin{align*}
	2~\E{S_\theta(X_j,Y)} &= \E{\mathbf{1}_{(\theta < X_j)}(\theta - Y)} = \E{\EC{\mathbf{1}_{(\theta < W + \eta_j)}(\theta - W - \varepsilon)}{\mathcal{F}}}\\
	&= \E{\EC{\mathbf{1}_{(\theta - W< \eta_j)}}{\mathcal{F}}\EC{(\theta - W - \varepsilon)}{\mathcal{F}}} = \E{(1-F_j^\mathcal{F}(\theta - W))(\theta - W)}. 
	\end{align*}
	Hence, \eqref{eq:sign-change-condition} implies
	\begin{align*}
	\E{S_\theta(X_B,Y)} &- \E{S_\theta(X_A,Y)} = \frac{1}{2}\E{\left(F_A^{\mathcal{F}}(\theta - W) - F_B^{\mathcal{F}}(\theta - W)\right)(\theta - W)} \ge 0.\qedhere
	\end{align*}

\end{proof}

\begin{proof}[Proof of Proposition \ref{prop:testable}]

	Parts (a) and (b) are immediate given the setup of Theorem \ref{thm:sign-condition}. Regarding (c), we have the following inequality for any strictly increasing function $\phi$:
	\begin{align}
	{\EC{\phi(|\eta_B|)}{\mathcal{F}}}
	&= {\int_{0}^\infty \mathbb{P}(\phi(|\eta_B|) > w|\mathcal{F})~dw }\nonumber\\
	&= {\int_{0}^\infty \left(\mathbb{P}(\eta_B < -\phi^{-1}(w)|\mathcal{F})+ \mathbb{P}(\eta_B > \phi^{-1}(w)|\mathcal{F})\right)~dw}\nonumber\\ 
	&\ge {\int_{0}^\infty \left(\mathbb{P}(\eta_A < -\phi^{-1}(w)|\mathcal{F})+ \mathbb{P}(\eta_A > \phi^{-1}(w)|\mathcal{F})\right)~dw} = \E{\phi(|\eta_A|)|\mathcal{F}},\label{ineq2}
	\end{align}
	where the inequality follows from \eqref{eq:sign-change-condition} in the proof of Theorem \ref{thm:sign-condition}.
	
	Now let $W = \EC{Y}{\mathcal{F}},$ such that $(X_j)^{2k} = (W + \eta_j)^{2k}$. For terms of the form $W^c\eta_j^d$, with $c$ and $d$ being odd integers, it holds that $\E{W^c\eta_j^d} = \E{W^c\EC{\eta_j^d}{\mathcal{F}}} = 0,$ where the last equality follows from symmetry of $F_j^{\mathcal{F}}$ around zero. For terms of the form $W^c\eta_j^d$, with $c$ and $d$ being even integers, it holds that $\E{W^c \eta_B^d} = \E{W^c~\EC{\eta_B^d}{\mathcal{F}}} \ge \E{W^c~\EC{\eta_A^d}{\mathcal{F}}},$ where the inequality follows from \eqref{ineq2}. Part (b) of Theorem \ref{thm:sign-condition} then follows from the binomial theorem.
\end{proof}

\end{document}